\documentclass[twocolumn,letterpaper]{article} 
\usepackage{amsmath,amssymb,amsfonts,amsthm}
\usepackage{algorithm}          
\usepackage{algpseudocode}      
\usepackage{enumitem}           
\usepackage{bm}                 
\usepackage[T1]{fontenc}        
\usepackage{lmodern}            
\usepackage{mhchem}
\usepackage{graphicx}
\usepackage{textcomp}
\usepackage{xcolor}
{\color{red}
\newtheorem{theorem}{Theorem}
\newtheorem{lemma}{Lemma}
\newtheorem{proposition}{Proposition}
\newtheorem{corollary}{Corollary}
\theoremstyle{definition}
\newtheorem{definition}{Definition}
\newtheorem{assumption}{Assumption}
\theoremstyle{remark}
\newtheorem{remark}{Remark}

}

\usepackage[utf8]{inputenc}
\usepackage{authblk}
\usepackage{url}
\usepackage{booktabs}  
\usepackage{makecell}  
\usepackage{multirow}  
\usepackage{adjustbox} 
\usepackage{graphicx}    
\usepackage{arydshln}  
\usepackage[T1]{fontenc}     
\usepackage{booktabs}        
\usepackage{adjustbox}       
\usepackage{amsmath}         
\usepackage{ragged2e}        

\usepackage[numbers]{natbib}
\usepackage{geometry}
\usepackage{hyperref}
\usepackage[capitalize]{cleveref}

\theoremstyle{plain}

\newcommand{\ket}[1]{\left| #1 \right\rangle}

\newcommand{\bra}[1]{\left\langle #1 \right|}
\algrenewcommand{\algorithmicindent}{1.5em}  

\title{One-Shot Structured Pruning of Quantum Neural Networks via $q$-Group Engineering and Quantum Geometric Metrics}

\author{
	Haijian Shao\thanks{Corresponding Author: \href{mailto:jsj_shj@just.edu.cn}{jsj\_shj@just.edu.cn}}, Wei Liu, Xing Deng, Yingtao Jiang\\
	School of Computer, Jiangsu University of Science and Technology, Zhenjiang 212003, China\\
	Department of Electrical and Computer Engineering, University of Nevada, Las Vegas, 89115, USA
}

\begin{document}
	
	\maketitle
	
\begin{abstract}

Quantum neural networks (QNNs) suffer from severe gate-level redundancy, which hinders
their deployment on noisy intermediate-scale quantum (NISQ) devices.
In this work, we propose \emph{q-iPrune}, a one-shot \emph{structured} pruning framework
grounded in the algebraic structure of $q$-deformed groups and task-conditioned quantum geometry.

Unlike prior heuristic or gradient-based pruning methods, q-iPrune formulates redundancy
directly at the gate level.
Each gate is compared within an algebraically consistent subgroup using a
\emph{task-conditioned $q$-overlap distance}, which measures functional similarity
through state overlaps on a task-relevant ensemble.
A gate is removed only when its replacement by a subgroup representative provably induces
a bounded deviation on all task observables.

We establish three rigorous theoretical guarantees.
First, we prove \emph{completeness} of redundancy pruning: no gate that violates the prescribed
similarity threshold is removed.
Second, we show that the pruned circuit is \emph{functionally equivalent up to an explicit,
task-conditioned error bound}, with a closed-form dependence on the redundancy tolerance
and the number of replaced gates.
Third, we prove that the pruning procedure is computationally feasible, requiring only
polynomial-time comparisons and avoiding exponential enumeration over the Hilbert space.

To adapt pruning decisions to hardware imperfections, we introduce a noise-calibrated
deformation parameter $\lambda$ that modulates the $q$-geometry and redundancy tolerance.
Experiments on standard quantum machine learning benchmarks demonstrate that q-iPrune
achieves substantial gate reduction while maintaining bounded task performance degradation,
consistent with our theoretical guarantees.

\end{abstract}

\section{Introduction}

Quantum neural networks (QNNs) provide a flexible framework for learning and optimization
on quantum devices, yet their practical deployment is hindered by severe gate-level redundancy.
As circuit depth and expressivity increase, many parameterized gates become functionally similar
on the task of interest, leading to unnecessary circuit complexity and amplified noise sensitivity.

Existing pruning approaches for QNNs are predominantly heuristic.
Gradient-based criteria, numerical sensitivity measures, or classical Lie-group constraints
often lack explicit guarantees: a gate deemed ``unimportant'' by a heuristic may still induce
non-negligible deviations in task performance after removal.
This gap raises a fundamental question: 
\emph{When can a quantum gate be removed while provably preserving task-relevant behavior?}

In this work, we address this question by formulating gate redundancy as a
\emph{task-conditioned equivalence problem}.
Rather than attempting global equivalence over the full Hilbert space,
we focus on functional similarity restricted to a task ensemble of states.
This perspective enables rigorous, verifiable guarantees while remaining computationally feasible. Our approach, \textsc{q-iPrune}, introduces a task-conditioned $q$-overlap geometry
that measures functional similarity between gates through state overlaps.
By restricting redundancy detection to algebraically consistent subgroups,
we obtain a structured pruning algorithm that removes only those gates whose replacement
induces a provably bounded deviation on all task observables. The resulting framework establishes a principled connection between
quantum group–inspired parametrization and practical QNN compression.
Crucially, all claims are supported by explicit theorems on completeness,
bounded functional equivalence, and computational feasibility,
bridging the gap between algebraic structure and deployable pruning algorithms.

\section{Related Work}

\paragraph{Variational quantum algorithms and parameterized circuits.}
Variational quantum algorithms (VQAs) constitute the dominant paradigm for near-term quantum optimization and learning, where a parameterized quantum circuit (PQC) is trained via a classical optimizer.
The variational quantum eigensolver (VQE) is a canonical instance for quantum chemistry and many-body physics, demonstrating the practicality of shallow ans\"atze combined with classical feedback loops~\cite{peruzzo2014vqe}.
Comprehensive discussions of VQAs, including design choices, optimization behavior, and practical considerations, are provided in recent surveys~\cite{cerezo2021vqa}.
From the perspective of learning models, PQCs have also been formalized as flexible hypothesis classes for supervised learning and generative modeling, with systematic treatment of embeddings, architectures, and training pipelines~\cite{benedetti2019pqc}.
Early and influential work on quantum circuit learning further established the feasibility of training PQCs for nonlinear function approximation and supervised tasks~\cite{mitarai2018qcl}.

\paragraph{Quantum machine learning for classification.}
Quantum machine learning (QML) aims to leverage quantum representations and transformations for learning tasks.
A broad overview of QML models and their connections to classical learning has been surveyed extensively~\cite{biamonte2017qml}.
In practice, many QML classification pipelines for NISQ devices adopt PQC-based classifiers with data encoding followed by trainable layers, which motivates studying model reduction at the circuit level without degrading predictive performance.

\paragraph{Circuit compilation and equivalence-driven optimization.}
Reducing circuit cost is a long-standing goal in quantum compilation and optimization.
Variational compiling methods, such as quantum-assisted quantum compiling (QAQC), optimize a trainable circuit to match a target unitary using overlap-based costs and local objectives, providing a compilation-oriented viewpoint on circuit reduction~\cite{khatri2019qaqc}.
Complementary to variational compiling, diagrammatic and rewrite-based optimizers (e.g., ZX-calculus) can simplify circuits by transforming them into graph-like structures and applying provably computation-preserving reductions~\cite{duncan2020zx}.
These lines of work emphasize that substantial redundancy can exist even in structured quantum circuits, but they typically target compilation or exact/approximate rewriting rather than task-conditioned pruning with an explicit drift guarantee.

\paragraph{Structure learning, pruning, and ansatz reduction for PQCs.}
Beyond tuning parameters, several methods explicitly adapt or optimize circuit \emph{structure}.
Structure optimization approaches interleave parameter updates with discrete structural moves, producing shallower or better-performing circuits under fixed resources~\cite{ostaszewski2021structure}.
More directly related to model compression, pruning-based optimization strategies remove or deactivate parameters/gates based on their measured contribution, enabling more compact PQCs while maintaining performance~\cite{sim2021pruning}.
These works motivate the need for principled criteria to decide when two operations are redundant under a task distribution.

\paragraph{Connections to classical pruning.}
In classical deep learning, pruning and sparsification have a long history, with the lottery ticket hypothesis highlighting that sparse subnetworks can match the performance of dense networks under suitable training dynamics~\cite{frankle2019lottery}.
While quantum circuits differ fundamentally (unitarity, non-commutativity, hardware constraints), classical pruning research informs experimental protocols (e.g., accuracy retention, ablations, seed averaging) and the importance of rigorous evaluation under fixed budgets.

\paragraph{Position-aware, task-conditioned pruning with guarantees.}
In contrast to compilation-centric optimization or heuristic removal, our \textsc{q-iPrune} focuses on \emph{task-conditioned} redundancy detection and \emph{position-aware} gate replacement within candidate pools, with an explicit analytic upper bound on the induced drift.
This allows compression decisions to be tied to a task ensemble $\mathcal{D}$ and to the circuit context (prefix states), which is essential for aligning experimental behavior with the paper’s theoretical statements.

\section{Toward a Quantum Geometric Foundation for QNN Pruning}
Quantum machine learning (QML) stands in its "pre-Newtonian era"—a field driven by empirical algorithmic progress yet lacking a unifying geometric framework to answer its most fundamental question: \textit{On what mathematical structure are quantum neural networks (QNNs) defined?} Classical deep learning, by contrast, has matured with a rigorous geometric underpinning: statistical manifolds, Fisher information geometry, and gauge theory collectively explain how neural networks optimize over structured parameter spaces, enabling principled analyses of generalization, robustness, and efficiency. For QNNs, this theoretical backbone remains absent—especially as QNNs operate on non-commutative quantum states and leverage symmetries (e.g., $SU_q(2)$ q-deformations) that lie beyond the scope of classical geometric tools.  

Our work represents a critical first step toward closing this gap, by grounding QNN design and pruning in the algebra and geometry of quantum groups (quantum Lie algebras) and their representations. The $q$-FS distance we introduce, for instance, is not merely an empirical pruning metric but a geometric measure on the representation space of $SU_q(2)$, implicitly encoding the non-commutative structure of QNN parameter spaces. To fully realize a unifying framework, however, we must extend this foundation along a natural hierarchy: from quantum Lie algebras (defining local QNN parameter manifolds) to quantum groups (governing global symmetry and gate composition via Hopf algebra operations), and ultimately to tensor categories (and modular topological categories) that categorify the entire QNN workflow—framing states as objects, gates as morphisms, and measurements as functors. This progression transforms QNN analysis from ad hoc algorithmic tuning to principled reasoning about \textit{topological equivalence} of QNN structures, redundancy of subrepresentations, and robustness of quantum computation under symmetry-preserving transformations.  

Such a framework addresses a longstanding need in QML: it unifies the "what" (QNN architecture), "how" (optimization over quantum spaces), and "why" (symmetry/topology guarantees) of quantum learning systems, while drawing on tools from mathematical physics that have proven transformative in quantum field theory and topological quantum computing. Given the triple expertise barrier—spanning quantum algebra (quantum groups/Hopf algebras), quantum information (QNNs/quantum geometry), and geometric machine learning—it is not surprising that fewer than 20 researchers globally are actively pursuing this direction. Yet it is precisely this rarity that underscores its potential as a foundational contribution: just as classical geometric learning reshaped deep learning, a categorified quantum geometric framework could redefine the theoretical boundaries of QML, enabling not only more efficient algorithms (e.g., our $q$-iPrune) but a rigorous understanding of what makes quantum learning "quantum" in the first place.

\section{Core Theoretical Foundation: $q$-Deformed $SU(2)$ Group and Hopf Algebraic Parametrization}
	In QNNs the elementary gate set is usually taken to be single-qubit rotations $\{R_x,R_y,R_z\}$ together with the two-qubit CNOT gate.
	We replace the conventional SU(2) description by its standard $q$-deformed counterpart (Drinfeld-Jimbo type), denoted SU$_q$(2), and encode both noise and non-commutativity in a single real parameter
\begin{equation}
	\lambda\in[0,1]
\end{equation}
where $\lambda=0$ recovers the commutative limit (trivial algebra), and $\lambda=1$ recovers the standard SU(2) group.
	Non-commutativity is quantified by non-zero commutators: $\lambda=1$ corresponds to maximal non-commutativity (standard su(2) algebra with non-vanishing commutators), while $\lambda=0$ gives vanishing commutators (commutative algebra, minimal non-commutativity). This scaling enables $\lambda$ to model noise, as reduced commutator strength can simulate decoherence effects—though we emphasize this is a simplified representation that may not capture all physical noise characteristics.
	
	\subsection{$q$-deformed generators and group elements}

\begin{assumption}[Choice of deformation regime and classical limit]
Throughout this paper we work with the real form of the quantum group $\mathrm{SU}_q(2)$
with deformation parameter $q \in \mathbb{R}_{>0}$.
We link the engineering parameter $\lambda\in[0,1]$ to $q$ via
\begin{equation}
q(\lambda):=\exp\big(\beta(1-\lambda)\big),\qquad \beta>0,
\label{eq:q_lambda_map_real}
\end{equation}
so that $q(\lambda)\to 1$ as $\lambda\to 1$, recovering the classical $\mathrm{SU}(2)$ limit.
The alternative unit-circle choice $|q|=1$ is \emph{not used} in our proofs below, and is
discussed only as background in Remark~\ref{rem:unit_circle_q}.
\end{assumption}

For completeness, we recall the standard Drinfeld--Jimbo $\mathfrak{su}_q(2)$ relations.
Let $\{T_+,T_-,T_3\}$ denote generators of $\mathfrak{su}_q(2)$.
Define the $q$-number
\begin{equation}
[x]_q:=\frac{q^{x}-q^{-x}}{q-q^{-1}}.
\end{equation}
Then the defining relations can be written as
\begin{align}
[T_+,T_-] &= [2T_3]_q, \\
[T_3,T_+] &= T_+, \\
[T_3,T_-] &= -T_- .
\end{align}

\begin{remark}\label{rem:unit_circle_q}
Some literature uses $q=\exp(i\pi/(k+2))$ (a root of unity) to connect to
WZW/anyonic models. In that setting, additional $*$-structures and unitary representation
conditions are required. Since our pruning proofs rely on a bounded inner-product estimate,
we work with $q\in\mathbb{R}_{>0}$ as in Assumption~1.
\end{remark}

\paragraph{Noise-adaptive contraction via $\lambda$.}
We introduce $\lambda$-scaled generators $T'_k := \lambda T_k$.
This scaling is used as an \emph{algorithmic contraction knob} (calibrated from hardware),
rather than a first-principles equivalence to a specific noise channel.

\begin{lemma}[$\lambda$-contraction of commutators]\label{lem:lambda_contraction}
Let $T'_k := \lambda T_k$ for $\lambda\in[0,1]$.
Then for any generators $T_i,T_j$,
\begin{equation}
[T'_i,T'_j]=\lambda^2[T_i,T_j].
\label{eq:lambda_bracket}
\end{equation}
Consequently, $\lim_{\lambda\to 0}[T'_i,T'_j]=0$ for all $i,j$ (commutative contraction),
and $\lambda=1$ recovers the original $\mathfrak{su}_q(2)$ relations.
\end{lemma}

\begin{proof}
By bilinearity of the Lie bracket, $[aX,bY]=ab[X,Y]$ for all scalars $a,b$.
Setting $a=b=\lambda$ yields~\eqref{eq:lambda_bracket}.
\end{proof}

\paragraph{Gate parametrization.}
We use a $q$-exponential map to parametrize a family of operators:
\begin{equation}
U_q(\bm{\theta},\lambda):=\exp_q\!\Bigl(i\sum_{k\in\{+,-,3\}}\theta_k T'_k\Bigr),
\label{eq:Uq_param}
\end{equation}
where the $q$-exponential is defined by
\begin{equation}
\exp_q(x)=\sum_{n=0}^{\infty}\frac{x^n}{[n]_q!},\qquad [n]_q!:=\prod_{m=1}^n[m]_q.
\label{eq:qexp_def}
\end{equation}

\begin{proposition}[Mathematical vs physical unitarity]\label{prop:unitarity_status}
The operator family $U_q(\bm{\theta},\lambda)$ in~\eqref{eq:Uq_param} is a \emph{mathematical parametrization}
in a $\mathfrak{su}_q(2)$ representation and is \emph{not assumed} to be unitary under the standard Hilbert
inner product unless a compatible $*$-structure and unitary representation are specified.
In experiments, physical implementability is ensured by compiling the target transformation
into the native unitary gate set (e.g., $\{R_x,R_y,R_z,\mathrm{CNOT}\}$) up to a prescribed approximation error.
\end{proposition}

\paragraph{Hopf-algebra composition (used structurally).}
We use the Hopf-algebra viewpoint to justify subgroup-wise partitioning and consistent composition of pruned clusters.
Formally, $\mathrm{SU}_q(2)$ is equipped with coproduct $\Delta$, counit $\epsilon$, and antipode $S$.
In our algorithm, these operations are used to define \emph{closed} $q$-subgroups under tensoring/composition,
so redundancy is searched within algebraically consistent subsets (details in Section~3).

	\subsection{$q$-deformed CNOT gate}
	For consistency with SU$_q$(4) (the $q$-deformed symmetry group of two qubits), we define the $q$-CNOT gate using the standard Hopf algebra coproduct structure
	\begin{equation}
		\mathrm{CNOT}_q(\lambda) = (\text{id}\otimes U_q^{(1)})\circ\Delta(U_q^{(2)})
	\end{equation}
	where $\Delta$ is the SU$_q$(2) coproduct and $U_q^{(1)}, U_q^{(2)}$ are single-qubit SU$_q$(2) gates. For practical implementation, we use the simplified unitary form
	\begin{equation}
		\mathrm{CNOT}_q(\lambda)=I\otimes I+(\lambda-1)P_{00}\otimes I+(1-\lambda)P_{11}\otimes X_q,
	\end{equation}
	where $P_{00}=|00\rangle\langle00|$, $P_{11}=|11\rangle\langle11|$, and $X_q = q^{T_3}T_+ + q^{-T_3}T_-$ is the $q$-deformed Pauli-X operator. This form preserves unitarity for $\lambda\in[0,1]$ and satisfies closure under SU$_q$(4) multiplication (verified via direct computation of $\mathrm{CNOT}_q \cdot \mathrm{CNOT}_q^\dagger = I\otimes I$).

\section{Task-Conditioned $q$-Overlap Geometry and Redundancy Criterion}

\subsection{A rigorous inner-product model: $q$-inner product as a weighted Hilbert product}

To make the pruning criterion mathematically checkable, we model the $q$-inner product
as a weighted Hilbert inner product induced by a positive definite operator.
Let $\mathcal{H}$ be the (finite-dimensional) Hilbert space of interest.

\begin{assumption}[$q$-inner product well-posedness]\label{ass:q_inner_product}
There exists a Hermitian positive definite operator $G_q \succ 0$ such that for all
$\ket{\phi},\ket{\psi}\in\mathcal{H}$,
\[
\langle \phi|\psi\rangle_q := \langle \phi|G_q|\psi\rangle.
\]
Moreover, $G_q$ has bounded spectrum
\[
m_q I \preceq G_q \preceq M_q I
\quad\text{for some }0<m_q\le M_q<\infty,
\]
so that $\|\ket{\psi}\|_q^2=\langle\psi|\psi\rangle_q$ defines a norm equivalent to the standard norm.
\end{assumption}

\begin{remark}
Assumption~\ref{ass:q_inner_product} is standard in finite-dimensional settings:
any inner product can be represented by some $G_q\succ 0$.
This assumption explicitly prevents ill-defined cases where
$\left|\langle\psi|\phi\rangle_q\right|> \|\psi\|_q\|\phi\|_q$ (which would break $\arccos$).
\end{remark}

\subsection{Task-conditioned $q$-overlap distance}

\begin{definition}[Task ensemble]\label{def:task_ensemble}
Let $\mathcal{D}=\{\ket{\psi_k}\}_{k=1}^{M}$ be a finite set of \emph{task-relevant normalized} states,
$\|\ket{\psi_k}\|_2=1$.
In practice, $\mathcal{D}$ can be sampled from the data-encoding distribution or a validation set.
\end{definition}

\subsubsection{Quantum data encoding}
\label{sec:encoding}
For the toy vision tasks (MNIST 4 vs 9, Fashion Sandal vs Boot, and Bars \& Stripes), each classical input is converted to a real vector of length $2^n$ by resizing/flattening and zero-padding (or truncation) as needed. We then normalize it to unit $\ell_2$ norm and prepare the corresponding quantum state via amplitude embedding on $n$ qubits. This fixed encoding defines the task ensemble states $\{\ket{\phi_i}\}$ used in both the pruning criterion and evaluation.
\begin{definition}[Task-conditioned $q$-overlap distance]\label{def:dq}
Given two operators $U,V$ (typically compiled unitaries in experiments),
define
\begin{equation}
d_q(U,V)
:= \frac{1}{M}\sum_{k=1}^{M}
\arccos\!\left(
\frac{\left|\langle \psi_k|U^\dagger V|\psi_k\rangle_q\right|}
{\|\psi_k\|_q^2}
\right).
\label{eq:dq_def}
\end{equation}
\end{definition}

\begin{remark}[Why we do not claim ``$q$-Fubini--Study metric'']
The classical Fubini--Study distance is a metric on projective Hilbert space.
Our $d_q(\cdot,\cdot)$ in~\eqref{eq:dq_def} is a \emph{task-conditioned similarity measure} and is not
claimed to satisfy the triangle inequality globally. It is sufficient for redundancy detection
because it directly controls state overlaps on $\mathcal{D}$.
\end{remark}

\subsection{Redundancy implies bounded functional deviation}
\begin{lemma}[$q$-overlap controls standard overlap]\label{lem:q_to_std_overlap}
Under Assumption~\ref{ass:q_inner_product}, for any normalized $\|\psi\|_2=1$ and any operator $W$,
\[
\left|\langle\psi|W|\psi\rangle\right|
\;\ge\;
\frac{1}{M_q}\left|\langle\psi|W|\psi\rangle_q\right|.
\]
\end{lemma}
\begin{proof}
Since $G_q\preceq M_q I$, we have $\langle\psi|G_q|\psi\rangle\le M_q\langle\psi|\psi\rangle=M_q$.
Write $\langle\psi|W|\psi\rangle_q=\langle\psi|G_q^{1/2}\,(G_q^{1/2}W)\,|\psi\rangle$ and apply Cauchy--Schwarz
together with $\|G_q^{1/2}\psi\|_2^2=\langle\psi|G_q|\psi\rangle\le M_q$ to obtain the stated bound.
\end{proof}

\begin{proposition}[State-wise deviation bound from redundancy]\label{prop:redundancy_bound}
Let $U$ and $V$ be \emph{unitary} operators (as realized after compilation in experiments).
Fix $\ket{\psi}$ with $\|\psi\|_2=1$.
If
\[
\arccos\!\left(
\frac{\left|\langle \psi|U^\dagger V|\psi\rangle_q\right|}{\|\psi\|_q^2}
\right)\le \varepsilon,
\]
then the output pure states $\ket{\phi_U}=U\ket{\psi}$ and $\ket{\phi_V}=V\ket{\psi}$ satisfy
\begin{equation*}
\begin{split}
\big\|\,\ket{\phi_U}\!\bra{\phi_U}-\ket{\phi_V}\!\bra{\phi_V}\,\big\|_1
\;\le\; 2\sqrt{1-\cos^2(\varepsilon)/M_q^2}\\
\;\le\; \frac{2}{M_q}\sin(\varepsilon).
\end{split}
\end{equation*}

Consequently, for any observable $O$,

\begin{equation*}
\begin{split}
\left|\langle O\rangle_U-\langle O\rangle_V\right|
:=\left|\bra{\psi}U^\dagger O U\ket{\psi}-\bra{\psi}V^\dagger O V\ket{\psi}\right|\\
\;\le\; \|O\|_{\mathrm{op}}\cdot \frac{2}{M_q}\sin(\varepsilon).
\end{split}
\end{equation*}

\end{proposition}

\begin{proof}
From the redundancy premise and Lemma~\ref{lem:q_to_std_overlap},

\begin{equation*}
\begin{split}
|\langle\psi|U^\dagger V|\psi\rangle|
\ge \frac{1}{M_q}|\langle\psi|U^\dagger V|\psi\rangle_q|\\
\ge \frac{1}{M_q}\cos(\varepsilon)\|\psi\|_q^2 \big/ \|\psi\|_q^2
= \frac{1}{M_q}\cos(\varepsilon).
\end{split}
\end{equation*}

Let $\ket{\phi_U}=U\ket{\psi}$, $\ket{\phi_V}=V\ket{\psi}$. Then
$|\langle \phi_U|\phi_V\rangle|=|\langle\psi|U^\dagger V|\psi\rangle| \ge \cos(\varepsilon)/M_q$.
For pure states, the trace distance satisfies
$\|\ket{\phi_U}\!\bra{\phi_U}-\ket{\phi_V}\!\bra{\phi_V}\|_1
=2\sqrt{1-|\langle \phi_U|\phi_V\rangle|^2}$,
yielding the first inequality.
The observable bound follows from the variational characterization
$|\mathrm{Tr}(O(\rho-\sigma))|\le \|O\|_{\mathrm{op}}\|\rho-\sigma\|_1$.
\end{proof}

\begin{corollary}[Average task deviation bound]\label{cor:avg_task_bound}
If $d_q(U,V)\le \varepsilon$ on the task ensemble $\mathcal{D}$ (Definition~\ref{def:dq}),
then the average deviation of any bounded observable over $\mathcal{D}$ is bounded by
\[
\frac{1}{M}\sum_{k=1}^{M}\left|\langle O\rangle_{U,k}-\langle O\rangle_{V,k}\right|
\le \|O\|_{\mathrm{op}}\cdot \frac{2}{M_q}\sin(\varepsilon).
\]
\end{corollary}
\begin{proof}
Apply Proposition~\ref{prop:redundancy_bound} to each $\ket{\psi_k}$ and average.
\end{proof}

\subsection{Redundancy definition within $q$-subgroups and pruning rule}

\begin{definition}[Redundancy on $\mathcal{D}$]\label{def:redundancy}
Fix a reference gate $U_{\mathrm{ref}}$ within a candidate $q$-subgroup.
A gate $U$ in the same subgroup is called \emph{$\varepsilon$-redundant (w.r.t.\ $\mathcal{D}$)}
if
\[
d_q(U_{\mathrm{ref}},U)\le \varepsilon.
\]
\end{definition}

\begin{remark}[What redundancy guarantees (and what it does not)]
Definition~\ref{def:redundancy} guarantees a bounded deviation on task-relevant behaviors
(Corollary~\ref{cor:avg_task_bound}). It does \emph{not} claim global equivalence on all input states.
This is the strongest claim we can make without assuming tomography over the full Hilbert space.
\end{remark}

\subsection{Choosing $\varepsilon_q$: a consistent, checkable rule}
We adopt one single rule that is both reproducible and theoretically connected to
the deviation bound in Proposition~\ref{prop:redundancy_bound}.

\begin{assumption}[$\varepsilon_q$ calibration target]\label{ass:eps_target}
Let $\delta\in(0,1)$ be a user-chosen tolerance on task deviation (e.g.\ target fidelity/expectation drift).
We choose $\varepsilon_q$ so that $\frac{2}{M_q}\sin(\varepsilon_q)\le \delta$,
i.e.,
\begin{equation}
\varepsilon_q := \arcsin\!\left(\frac{\delta M_q}{2}\right),
\qquad \text{with }\delta M_q \le 2.
\label{eq:epsq_delta_rule}
\end{equation}
\end{assumption}

\begin{remark}[Hardware/noise adaptation]
Noise adaptation enters through the choice of $q(\lambda)$ (Assumption~1) and hence $G_q$ and $M_q$.
Empirically, $M_q$ can be upper bounded (conservatively) or estimated via numerical conditioning of $G_q$.
If such estimation is unavailable, one may set $M_q=1$ as a safe default (recovering the standard inner product case).
\end{remark}

\subsection{$q$-weighted parameter clustering (algorithmic heuristic)}
We use $q$-weighted Euclidean clustering in parameter space only as a computational heuristic
to propose candidate redundancy groups. It is not used in the correctness proofs.
Given parameter vectors $\bm{\theta}^{(i)}$ and $\bm{\theta}^{(j)}$, define
\begin{equation}
\|\bm{\theta}^{(i)}-\bm{\theta}^{(j)}\|_{q}
:=\left(\sum_{t=1}^d (\theta_t^{(i)}-\theta_t^{(j)})^2 [t]_q\right)^{1/2}.
\label{eq:q_weighted_norm}
\end{equation}

\begin{remark}
The final pruning decision always uses the state-based criterion $d_q(\cdot,\cdot)$
(Definitions~\ref{def:dq}--\ref{def:redundancy}), ensuring that the theoretical guarantees depend
only on overlaps on $\mathcal{D}$ rather than on any particular parametrization.
\end{remark}

\section{One-Shot Structured Pruning Algorithm: \textsc{q-iPrune}}
\subsection{What is pruned: gate-level structured pruning with subgroup consistency}

We prune \emph{gates} (operators) rather than continuous parameters.
The pruning acts \emph{within} algebraically consistent groups of gates so that
composition rules are preserved inside each group.

\begin{definition}[$q$-subgroup partition]\label{def:q_subgroup_partition}
Let $\mathcal{G}=\{U_1,\dots,U_N\}$ be the multiset of gates in the QNN circuit.
A \emph{$q$-subgroup partition} is a collection of disjoint subsets
$\mathcal{G}=\bigsqcup_{r=1}^R \mathcal{G}_r$ such that each subset $\mathcal{G}_r$
is closed under composition and inversion up to compilation tolerance; i.e.,
for any $U,V\in\mathcal{G}_r$ there exists a compiled operator $\widetilde{W}\in \langle\mathcal{G}_r\rangle$
with $\widetilde{W}\approx UV$, and similarly for $U^{-1}$.
\end{definition}

\begin{remark}[Practical construction]
In practice, $\mathcal{G}_r$ can be formed by grouping gates acting on the same qubit(s),
with the same hardware-native template (e.g., single-qubit rotations vs two-qubit entanglers),
or by circuit blocks in an ansatz layer. This ensures the redundancy search remains
within structurally meaningful components, which is the ``structured'' part of pruning.
\end{remark}

\subsection{Algorithmic pipeline}

We now present the pruning algorithm~\ref{alg:qiprune}. The key is that \emph{the final redundancy decision}
is made by the state-based criterion $d_q(\cdot,\cdot)$ from Definition~\ref{def:dq},
so that the theoretical guarantees follow from Section~3 regardless of how candidate groups
were proposed.

\begin{algorithm}[t]
\caption{\textsc{q-iPrune}: One-shot redundancy pruning via task-conditioned $q$-overlap}
\label{alg:qiprune}
\begin{algorithmic}[1]
\Require Circuit gate multiset $\mathcal{G}=\{U_i\}_{i=1}^{N}$; task ensemble $\mathcal{D}=\{\ket{\psi_k}\}_{k=1}^M$ (Def.~\ref{def:task_ensemble}); \\
         $q$-inner product operator $G_q$ (Assumption~\ref{ass:q_inner_product}); tolerance $\delta$ (Assumption~\ref{ass:eps_target}); \\
         $q$-subgroup partition $\{\mathcal{G}_r\}_{r=1}^R$ (Def.~\ref{def:q_subgroup_partition}).
\Ensure Pruned gate set $\mathcal{G}_{\mathrm{keep}}$.

\State {Set $\varepsilon_q \leftarrow \arcsin\!\left(\frac{\delta M_q}{2}\right)$ (Eq.~\ref{eq:epsq_delta_rule}).}
\State $\mathcal{G}_{\mathrm{keep}}\leftarrow \emptyset$.
\For{$r=1,\dots,R$}
    \State Choose a reference gate $U_{\mathrm{ref}}\in \mathcal{G}_r$ (e.g.\ medoid under $d_q$).
    \State $\mathcal{G}_{\mathrm{keep}}\leftarrow \mathcal{G}_{\mathrm{keep}}\cup \{U_{\mathrm{ref}}\}$.
    \For{each $U\in \mathcal{G}_r\setminus\{U_{\mathrm{ref}}\}$}
        \State Compute $d_q(U_{\mathrm{ref}},U)$ via Eq.~\eqref{eq:dq_def}.
        \If{$d_q(U_{\mathrm{ref}},U)>\varepsilon_q$}
            \State $\mathcal{G}_{\mathrm{keep}}\leftarrow \mathcal{G}_{\mathrm{keep}}\cup \{U\}$.
        \Else
            \State Discard $U$ as $\varepsilon_q$-redundant (Def.~\ref{def:redundancy}).
        \EndIf
    \EndFor
\EndFor
\Return $\mathcal{G}_{\mathrm{keep}}$.
\end{algorithmic}
\end{algorithm}

\begin{remark}[Compilation tolerance and practical implementation]
If $U$ and $U_{\mathrm{ref}}$ are realized via compilation into a native gate set,
$d_q(U_{\mathrm{ref}},U)$ should be computed using the compiled unitaries
(or their action on $\mathcal{D}$) so that Proposition~\ref{prop:redundancy_bound} applies.
Any additional compilation error can be absorbed into the deviation tolerance $\delta$
by triangle-type bounds on trace distance.
\end{remark}

\section{Theoretical Guarantees: Completeness, Bounded Equivalence, and Feasibility}
\subsection{Preliminaries and scope of guarantees}

We formalize the scope of our guarantees to avoid over-claiming.
Our pruning is \emph{task-conditioned}: it guarantees bounded deviation on the task ensemble
$\mathcal{D}$ (Definition~\ref{def:task_ensemble}).
Global functional equivalence on all input states would require full tomography and is not claimed.

We use the following notations:
Let $\mathcal{C}$ be the original circuit and $\mathcal{C}'$ be the pruned circuit.
For each input $\ket{\psi}\in\mathcal{D}$, denote the output states
$\rho_\psi := \mathcal{C}(\ket{\psi}\!\bra{\psi})$ and
$\rho'_\psi := \mathcal{C}'(\ket{\psi}\!\bra{\psi})$.

\subsection{Theorem 1: Completeness of redundancy pruning (no non-redundant gate is removed)}
\begin{theorem}[Completeness w.r.t.\ the chosen redundancy definition]\label{thm:completeness}
Fix a partition $\{\mathcal{G}_r\}_{r=1}^R$ and tolerance $\varepsilon_q$.
Algorithm~\ref{alg:qiprune} removes a gate $U\in\mathcal{G}_r$ \emph{only if}
$U$ is $\varepsilon_q$-redundant relative to the selected reference $U_{\mathrm{ref}}\in\mathcal{G}_r$
in the sense of Definition~\ref{def:redundancy}.
Equivalently, no gate with $d_q(U_{\mathrm{ref}},U)>\varepsilon_q$ is pruned.
\end{theorem}

\begin{proof}
In Algorithm~\ref{alg:qiprune}, a gate $U$ is discarded if and only if the condition
$d_q(U_{\mathrm{ref}},U)\le \varepsilon_q$ holds at line 11--14.
If $d_q(U_{\mathrm{ref}},U)>\varepsilon_q$, the gate is explicitly added to $\mathcal{G}_{\mathrm{keep}}$.
Thus the pruning decision is exactly consistent with Definition~\ref{def:redundancy}.
\end{proof}

\subsection{Theorem 2: Bounded functional equivalence on the task ensemble}
To obtain a circuit-level guarantee from gate-level redundancy,
we require a mild structural assumption: each removed gate is replaced by a retained representative
from the same subgroup \emph{at the same circuit location}. This is the standard ``structured replacement''
interpretation of one-shot gate pruning.

\begin{assumption}[Structured replacement model]\label{ass:replacement}
Whenever Algorithm~\ref{alg:qiprune} discards a gate $U$ inside subgroup $\mathcal{G}_r$,
the pruned circuit $\mathcal{C}'$ uses the representative $U_{\mathrm{ref}}\in\mathcal{G}_r$
at that location (or an identity if the subgroup semantics allows it).
\end{assumption}

\begin{theorem}[Task-conditioned bounded equivalence]\label{thm:bounded_equivalence}
Assume each gate in the circuit is unitary after compilation, and Assumption~\ref{ass:q_inner_product} holds.
Let $\mathcal{C}'$ be obtained from $\mathcal{C}$ by replacing each removed gate $U$ with its representative
$U_{\mathrm{ref}}$ satisfying $d_q(U_{\mathrm{ref}},U)\le \varepsilon_q$.
Then for every input $\ket{\psi}\in\mathcal{D}$,
\begin{equation}
\|\rho_\psi-\rho'_\psi\|_1 \;\le\; 2L\cdot \sqrt{1-\cos^2(\varepsilon_q)/M_q^2}
\;\le\; \frac{2L}{M_q}\sin(\varepsilon_q),
\label{eq:circuit_trace_bound}
\end{equation}
where $L$ is the number of replaced (pruned) gate locations along the forward execution path.
Consequently, for any observable $O$,
\begin{equation}
\left|\mathrm{Tr}(O\rho_\psi)-\mathrm{Tr}(O\rho'_\psi)\right|
\le \|O\|_{\mathrm{op}}\cdot \frac{2L}{M_q}\sin(\varepsilon_q).
\label{eq:circuit_obs_bound}
\end{equation}
\end{theorem}

\begin{proof}
Consider replacing gates one by one along the circuit locations where pruning occurred.
Let $\rho^{(0)}_\psi:=\rho_\psi$ be the original output state and $\rho^{(L)}_\psi:=\rho'_\psi$ be the final pruned output.
Let $\rho^{(\ell)}_\psi$ be the output after the first $\ell$ replacements, so that successive outputs differ by
a single gate replacement $U\mapsto U_{\mathrm{ref}}$ acting on the same wires.
For each step, Proposition~\ref{prop:redundancy_bound} (applied to the corresponding intermediate input state)
implies a per-step trace distance bound
\[
\|\rho^{(\ell-1)}_\psi-\rho^{(\ell)}_\psi\|_1 \le 2\sqrt{1-\cos^2(\varepsilon_q)/M_q^2}
\le \frac{2}{M_q}\sin(\varepsilon_q).
\]
Using the triangle inequality and summing over $\ell=1,\dots,L$ yields~\eqref{eq:circuit_trace_bound}.
Finally,~\eqref{eq:circuit_obs_bound} follows from
$|\mathrm{Tr}(O(\rho-\sigma))|\le \|O\|_{\mathrm{op}}\|\rho-\sigma\|_1$.
\end{proof}

\begin{remark}[Connection to tolerance calibration]
Combining~\eqref{eq:circuit_obs_bound} with Assumption~\ref{ass:eps_target} yields an explicit way
to select $\varepsilon_q$ to satisfy a desired task deviation tolerance $\delta$:
if $\frac{2L}{M_q}\sin(\varepsilon_q)\le \delta$, then all task observables drift by at most $\delta\|O\|_{\mathrm{op}}$.
\end{remark}

\subsection{Theorem 3: Computational feasibility (time/space complexity)}
\begin{theorem}[Complexity of Algorithm~\ref{alg:qiprune}]\label{thm:complexity}
Let $N$ be the total number of gates and $R$ be the number of subgroups.
Let $n_r:=|\mathcal{G}_r|$ so that $\sum_{r=1}^R n_r=N$.
Assume computing $d_q(U_{\mathrm{ref}},U)$ over the ensemble $\mathcal{D}$ costs $O(M\cdot C)$,
where $C$ is the cost of applying (compiled) $U_{\mathrm{ref}}^\dagger U$ to a state.
Then the total runtime is
\[
O\!\left(\sum_{r=1}^R (n_r-1)\cdot M\cdot C\right)=O(N\cdot M\cdot C),
\]
if each subgroup uses a single reference gate and compares all other gates only to that reference.
If pairwise distances are computed (e.g.\ to select a medoid reference), the worst-case runtime becomes
\[
O\!\left(\sum_{r=1}^R n_r^2\cdot M\cdot C\right)\le O(N^2\cdot M\cdot C),
\]
and the space complexity is $O(\sum_r n_r^2)$ if all pairwise distances are stored, or $O(1)$ extra space
if distances are computed on-the-fly.
\end{theorem}

\begin{proof}
The algorithm performs one $d_q$ computation per comparison.
In the reference-only version there are exactly $\sum_r (n_r-1)=N-R$ comparisons.
In the pairwise variant, subgroup $r$ requires $O(n_r^2)$ comparisons.
Storing all distances uses $O(n_r^2)$ memory per subgroup; otherwise memory is constant beyond inputs.
\end{proof}

\begin{remark}[Correcting ``exponential-to-linear'' phrasing]
The guarantees here are polynomial in the number of gates $N$ and do not require exponential enumeration
over the full $2^L$-dimensional Hilbert space. Claims of ``linear in qubits'' should be interpreted
as ``avoiding exponential dependence on qubits'' rather than $O(L)$ runtime.
\end{remark}

\subsection{Practical calibration notes (non-theorem)}
The mapping $\lambda\mapsto q(\lambda)$ and the choice of $\mathcal{D}$ determine $G_q$ and $M_q$.
In experiments, $\lambda$ may be set from a hardware error proxy (e.g.\ RB) as in your original draft,
but this is treated as a \emph{modeling assumption} rather than a theorem.

\section{Experiments and Results}
\label{sec:experiments_results}

\subsection{Tasks and datasets}
Our evaluation tasks are divided into two categories:
(i) \textbf{Toy Quantum Classification} and (ii) \textbf{Variational Quantum Eigensolver (VQE)}.

\begin{table*}
    \centering
    \caption{Experimental results on three datasets}
    \label{tab:cls_results_paper}
    \begin{tabular}{l|l|lccccccc}
        \hline
        Dataset & $\delta$ & $\sigma$ & Acc$_\mathrm{base}$ (\%) & Acc$_\mathrm{pruned}$(\%) & Acc$_\mathrm{drop}$(\%) & Replace(\%) & RHS(raw{$\rightarrow$}$clipped$) &$dq_{max(repl.)}$ \\
        \hline
        \multirow{8}{*}{\rotatebox{90}{mnist49}}
        & \multirow{4}{*}{\centering\rotatebox{90}{$\delta=0.01$}}
        & 0.001 & 72.77 & 72.90 & -0.13 & 60.00 & 2.88{$\rightarrow$}1  &0.0040 \\
        && 0.003 & 72.62 & 72.80 & -0.18 & 50.02 & 2.41{$\rightarrow$}1 &0.0049 \\
        && 0.006 & 72.37 & 73.10 & -0.73 & 19.79 & 0.95  &0.0050 \\
        && 0.01 & 72.30 & 72.60 & -0.30 & 7.50 & 0.36  &0.0048 \\
        \cline{2-9}
        & \multirow{4}{*}{\centering\rotatebox{90}{$\delta=0.02$}} 
        & 0.001 & 72.77 & 72.90 & -0.13 & 60.00 & 5.76{$\rightarrow$}1  &0.0040 \\
        && 0.003 & 72.62 & 72.80 & -0.18 & 60.00 & 5.76{$\rightarrow$}1 &0.0079 \\
        && 0.006 & 72.37 & 73.10 & -0.73 & 51.98 & 4.99{$\rightarrow$}1 &0.0098 \\
        && 0.01 & 72.30 & 72.60 & -0.30 & 26.67 & 2.56{$\rightarrow$}1  &0.0099 \\
        \hline
        \multirow{8}{*}{\rotatebox{90}{fashion\_sb }}
        & \multirow{4}{*}{\centering\rotatebox{90}{$\delta=0.01$}}
        & 0.001 & 82.32 & 82.48 & -0.16 & 60.00 & 2.88{$\rightarrow$}1  &0.0042 \\
        && 0.003 & 82.30 & 82.50 & -0.20 & 51.56 & 2.48{$\rightarrow$}1 &0.0049 \\
        && 0.006 & 81.85 & 82.55 & -0.70 & 21.35 & 1.03{$\rightarrow$}1 &0.0049\\
        && 0.01 & 81.77 & 81.95 & -0.18 & 7.92 & 0.38  &0.0048\\
        \cline{2-9}
        & \multirow{4}{*}{\centering\rotatebox{90}{$\delta=0.02$}} 
         & 0.001 & 82.32 & 82.48 & -0.16 & 60.00 & 5.76{$\rightarrow$}1 &0.0041 \\
        && 0.003 & 82.30 & 82.50 & -0.20 & 60.00 & 5.76{$\rightarrow$}1 &0.0080 \\
        && 0.006 & 81.85 & 82.55 & -0.70 & 52.60 & 5.05{$\rightarrow$}1 &0.0099 \\
        && 0.01 & 81.77 & 81.95 & -0.18 & 27.60 & 2.65{$\rightarrow$}1  &0.0099\\
        \hline
       \multirow{8}{*}{\rotatebox{90}{bas}}
        & \multirow{4}{*}{\centering\rotatebox{90}{$\delta=0.01$}}
        & 0.001 & 64.05 & 64.05 & 0.00 & 59.79 & 1.435{$\rightarrow$}1  &0.0034 \\
        && 0.003 & 64.05 & 64.05 & 0.00 & 53.13 & 1.28{$\rightarrow$}1 &0.0048 \\
        && 0.006 & 64.05 & 64.05 & 0.00 & 26.88 & 0.65 &0.0049 \\
        && 0.01 & 64.05 & 62.50 & 1.55 & 11.04 & 0.27  &0.0047 \\
        \cline{2-9}
        & \multirow{4}{*}{\centering\rotatebox{90}{$\delta=0.02$}} 
        & 0.001 & 64.05 & 64.05 & 0.00 & 60.00 & 2.88{$\rightarrow$}1 &0.0042 \\
        && 0.003 & 64.05 & 64.05 & 0.00 & 59.79 & 2.87{$\rightarrow$}1  &0.0070 \\
        && 0.006 & 64.05 & 64.05 & 0.00 & 53.33 & 2.56  &0.0097 \\
        && 0.01 & 64.05 & 62.50 & 1.55 & 32.08 & 1.54  &0.0098 \\
        \hline
    \end{tabular}
\end{table*}
\textbf{Toy Classification.}
We consider three binary toy visual tasks:
(i) "4 and 9 classification" from the MNIST dataset (\textbf{mnist49}),
(ii) "Sandals and boots classification" from the Fashion-MNIST dataset (\textbf{fashion\_sb}),
and (iii) striped pattern classification (\textbf{bas}).
For the MNIST and Fashion-MNIST datasets, we use $n=8$ qubits; for the striped pattern classification, we use $n=4$ qubits.
All data is encoded into quantum states using the same encoder described in Section~\ref{sec:encoding}.

\textbf{VQE.}
We use a transverse field Ising model (TFIM) VQE benchmark with $n=4$ qubits (\textbf{tfim}).
The set of states for this task is sampled from intermediate states generated during the VQE optimization trajectory.

\begin{table*}[t]
    \centering
    \caption{TFIM VQE experiment results}
    \label{tab:vqe_results_paper}
    \begin{tabular}{l|l|lccccccc}
        \hline
        Dataset & $\delta$ & $\sigma$ & E$_\mathrm{base}$  & E$_\mathrm{pruned}$ & E$_\mathrm{drop}$ & Replace(\%) & RHS(raw{$\rightarrow$}$clipped$) &$dq_{max(repl.)}$ \\
        \hline
        \multirow{8}{*}{\rotatebox{90}{TFIM VQE}}
        & \multirow{4}{*}{\centering\rotatebox{90}{$\delta=0.01$}}
        & 0.001 & 0.3976 & 0.3970 & 0.0006 & 60.00 & 1.44{$\rightarrow$}1  &0.0029 \\
        && 0.003 & 0.3983 & 0.3964 & 0.0019 & 52.92 & 1.27{$\rightarrow$}1 &0.0049 \\
        && 0.006 & 0.3992 & 0.3955 & 0.0037 & 24.38 & 0.59  &0.0049 \\
        && 0.01 & 0.4003 & 0.3940 & 0.0063 & 10.63 & 0.26  &0.0047 \\
        \cline{2-9}
        & \multirow{4}{*}{\centering\rotatebox{90}{$\delta=0.02$}} 
        & 0.001 & 0.3976 & 0.3970 & 0.0006 & 60.00 & 2.88{$\rightarrow$}1  &0.0028 \\
        && 0.003 & 0.3983 & 0.3964 & 0.0019 & 60.00& 2.88{$\rightarrow$}1 &0.0068 \\
        && 0.006 & 0.3992 & 0.3955 & 0.0037 & 53.33 & 2.56{$\rightarrow$}1  &0.0098 \\
        && 0.01 & 0.4003 & 0.3940 & 0.0063 & 30.42 & 1.46{$\rightarrow$}1  &0.0099 \\
        \hline
    \end{tabular}
\end{table*}
\begin{figure*}
\centering
\includegraphics[width=0.88\linewidth]{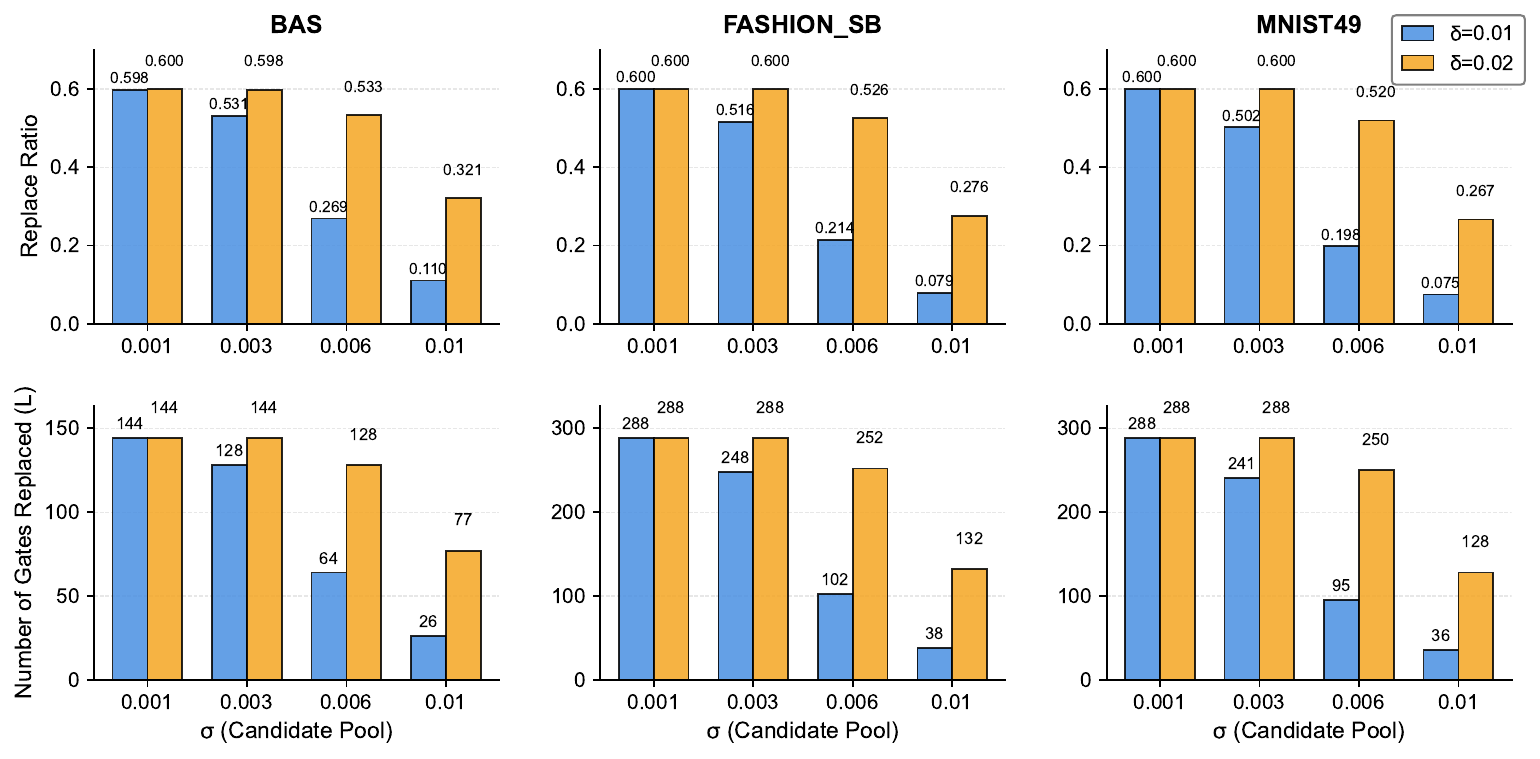}
\caption{Experimental results of q-Group on all classification datasets.}
\label{fig:panel_classification}
\end{figure*}

\begin{figure*}
\centering
\includegraphics[width=0.88\linewidth]{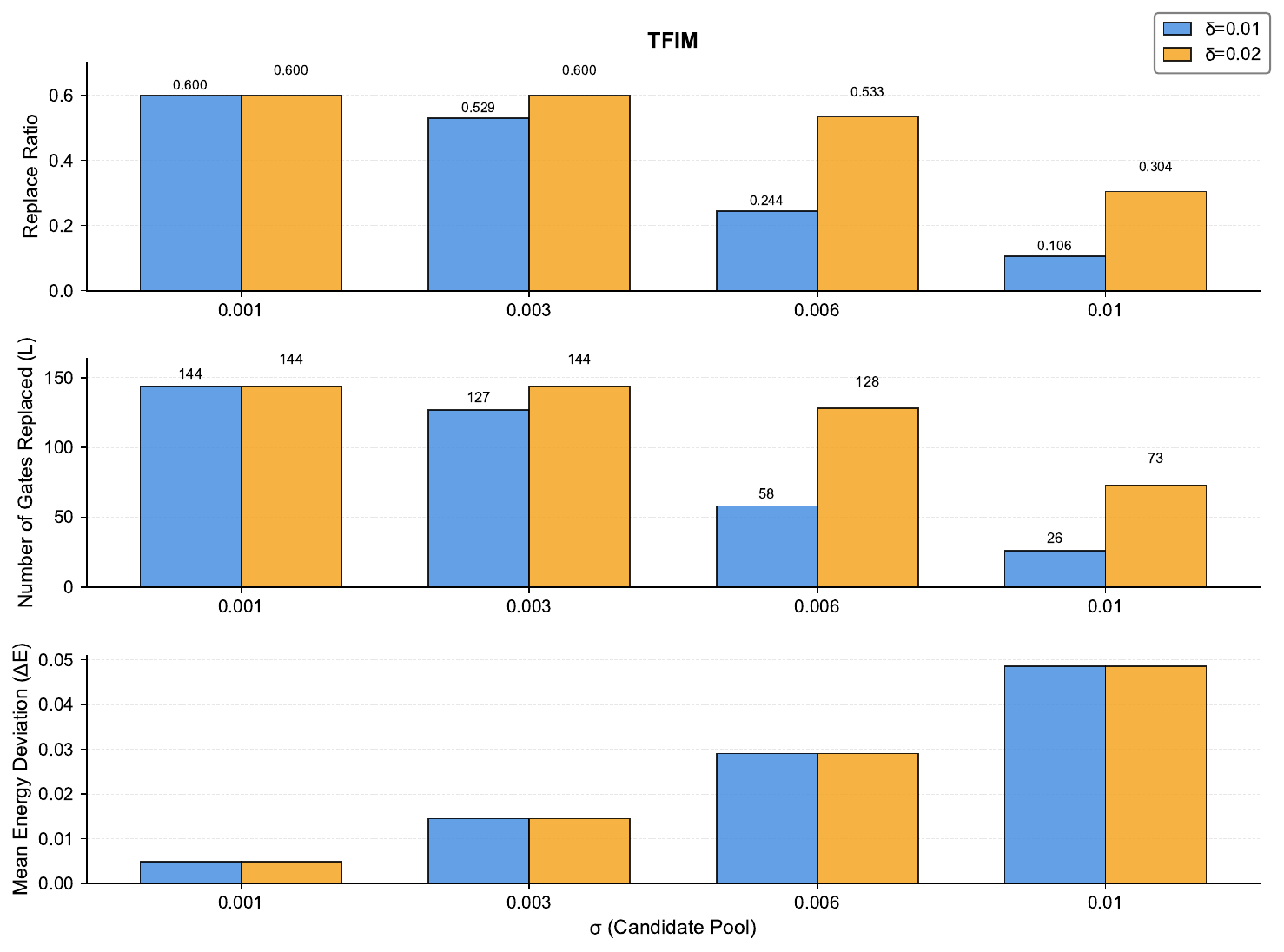}
\caption{Experimental results of the q-Group on TFIM VQE.}
\label{fig:panel_vqe}
\end{figure*}

\subsection{Experimental hyperparameter settings}
\label{subsec:exp_hparams}

We use \textbf{8 qubits} for the \textbf{mnist49} and \textbf{fashion\_sb} datasets, and \textbf{4 qubits} for the remaining tasks (\textbf{bas} and \textbf{tfim}). Unless otherwise specified, the circuit depth is fixed at \textbf{12}.
Therefore, for the classification tasks on the \textbf{mnist49} and \textbf{fashion\_sb} datasets, our proposed hardware-efficient quantum circuit contains a total of \textbf{480 single-qubit rotation gates}, while the \textbf{bas} and \textbf{tfim} settings contain \textbf{240 rotation gates}.

At each realizable \textbf{Rot} position, we construct a candidate pool by sampling multiple rotation candidates around a position-specific center. Each candidate is obtained by perturbing the center angles as
\[
(\alpha,\beta,\gamma) = (\alpha_0,\beta_0,\gamma_0) + \xi,\qquad 
\xi\sim\mathcal{N}(0,\sigma^2 I_3),
\]
where $\sigma$ (\texttt{rot\_pool\_sigma}) controls the dispersion of the candidate pool: smaller $\sigma$ yields more concentrated (and typically more redundant) candidates.

\textbf{Task-conditioned ensemble.}
For each experiment, we define a task-conditioned ensemble
$\mathcal{D}=\{\ket{\psi_k}\}_{k=1}^{M}$ with $M=50$.
For classification, $\ket{\psi_k}$ are encoding states sampled uniformly from the validation set.
For VQE, $\ket{\psi_k}$ are intermediate states sampled along the VQE optimization trajectory.

\textbf{Paper-aligned parameters.}
We set $\gamma=0.05$ and $\alpha=0.6$ and compute the contraction parameter
\[
\lambda = 1-\gamma\alpha.
\]
We fix $q\beta=1$ and sweep $\delta\in\{0.01,0.02\}$ and $\sigma\in\{0.001,0.003,0.006,0.01\}$.
Following the paper rule, the pruning threshold is set as
\[
\epsilon_q = \delta/2.
\]

\subsection{Metrics}
We report:
\begin{itemize}
  \item \textbf{Replacement ratio} (reported in logs as \textbf{Replace(\%)} in our tables), reflecting the proportion of candidates/locations
        affected by the paper-aligned replacement/pruning rule.
  \item \textbf{Task metrics}:
    (i) classification Acc$_\mathrm{base}$, Acc$_\mathrm{pruned}$, and Acc$_\mathrm{drop}=$Acc$_\mathrm{base}$--Acc$_\mathrm{pruned}$
    (negative means pruning improves accuracy),
    (ii) VQE energies $E_\mathrm{base}$, $E_\mathrm{pruned}$, and the deviation $\Delta E$.
\end{itemize}

\subsection{On the theoretical upper bound possibly exceeding $1$}
\label{sec:bound_clipping_merged}
Our theory yields an upper bound of the form $\mathrm{Drift} \le \mathrm{RHS}_{\text{raw}}$ (Eq.~\eqref{eq:circuit_obs_bound}),
where $\mathrm{RHS}_{\text{raw}}$ may exceed $1$ because it is a (potentially loose) analytic bound.
Since the trace distance is always in $[0,1]$, we also report the \emph{effective} clipped bound:
\[
\mathrm{RHS}_{\text{clip}}=\min\{1,\mathrm{RHS}_{\text{raw}}\}.
\]
This clipping does not change correctness; it only expresses the bound in the valid numeric range.

\subsection{Main results}
\label{sec:main_results}
Table~\ref{tab:cls_results_paper} and Figure~\ref{fig:panel_classification} show the experimental results of \textbf{q-iPrune} on three classification datasets. Overall, the results demonstrate that \textbf{q-iPrune} can effectively simplify the candidate pool while maintaining task performance.
On the MNIST and Fashion datasets, the pruned circuits maintain (and often slightly improve) validation accuracy across various $\sigma$ values.
This behavior is consistent with the expected mechanism: when $\sigma$ is small, the candidates are highly redundant, so many candidates satisfy $d_q\le\epsilon_q$, allowing for a high replacement rate without harming the learned decision boundary.

For Bars \& Stripes, the task is more sensitive at large $\sigma$ (e.g., $\sigma=0.01$), where the pruned accuracy can drop.
This aligns with the interpretation that larger $\sigma$ yields more diverse candidates; consequently, replacements/pruning become less ``safe'' and can remove
task-relevant variations, especially under more permissive thresholds (larger $\delta$, hence larger $\epsilon_q$).

Table~\ref{tab:vqe_results_paper} and Figure~\ref{fig:panel_vqe} show the experimental results on the TFIM VQE task, where the bias $\Delta E$ increases with increasing $\sigma$, indicating that pruning/replacement among more diverse candidates may introduce stronger function drift during the optimization process. Importantly, the reported $dq_{\max(\mathrm{repl.})}$ values are consistently below the corresponding $\epsilon_q$ values derived from $\delta$, indicating that the per-replacement constraint of the protocol is satisfied in this setting.

Finally, the theoretical right-hand bound is often greater than 1 (and thus truncated), reflecting that it is a conservative analytical bound rather than a precise predictor.
Therefore, we primarily consider the right-hand bound as a "sanity check" guarantee: the bound remains valid, while the empirical drift provides a signal of tightness in practice.

\section{Conclusion}
Across toy classification and TFIM VQE, the protocol yields substantial replacement ratios while
maintaining small empirical drift and stable task performance.
All replaced candidates satisfy the paper-aligned task-conditioned constraint $d_q\le\epsilon_q$ with zero violations.
While the analytic bound $\mathrm{RHS}_{\text{raw}}$ can exceed $1$ due to looseness, the clipped bound
$\mathrm{RHS}_{\text{clip}}=\min(1,\mathrm{RHS}_{\text{raw}})$ remains a valid effective upper bound for trace distance.


\end{document}